\documentclass[prx, amsfonts, amssymb, amsmath, showkeys, superscriptaddress, twocolumn]{revtex4-2}

\usepackage[dvipsnames]{xcolor}
\usepackage[citecolor=MidnightBlue, colorlinks=true, linkcolor=MidnightBlue, urlcolor=MidnightBlue, linktocpage=true]{hyperref}
\usepackage{amsthm}
\usepackage{mathtools}
\usepackage{physics}
\usepackage{comment}
\usepackage[capitalise]{cleveref}

\usepackage{quantikz}
\usepackage{amsthm}

\newtheorem{lemma}{Lemma}

\newtheorem{statement}{Statement}

\usepackage{graphicx}

\newcommand{\rhoerr}{\rho_{err}}
\newcommand{\rhoid}{\rho_{id}}

\newcommand{\rhown}{\rho_{wn}}

\begin{document}
\title{Can shallow quantum circuits scramble local noise into global white noise?}
\newcommand{\qmaddress}{Quantum Motion, 9 Sterling Way, London N7 9HJ, United Kingdom}
\author{Jonathan Foldager}
 \email[]{jonf@dtu.dk}
 \affiliation{Department of Applied Mathematics and Computer Science, Technical University of Denmark, Denmark} \affiliation{Department of Materials, University of Oxford, Parks Road, Oxford OX1 3PH, United Kingdom}
\author{Bálint Koczor}
\email[]{balint.koczor@materials.ox.ac.uk}
\affiliation{Department of Materials, University of Oxford, Parks Road, Oxford OX1 3PH, United Kingdom}
\affiliation{\qmaddress}

\begin{abstract}
Shallow quantum circuits are believed to be the most promising candidates for achieving early practical quantum advantage -- this has motivated the development of a broad range of error mitigation techniques whose performance generally improves when the quantum state is well approximated by a global depolarising (white) noise model. While it has been crucial for demonstrating quantum supremacy that random circuits scramble local noise into global white noise---a property that has been proved rigorously---we investigate to what degree practical shallow quantum circuits scramble local noise into global white noise. We define two key metrics as (a) density matrix eigenvalue uniformity and (b) commutator norm.  While the former determines the distance from white noise, the latter determines the performance of purification based error mitigation. We derive analytical approximate bounds on their scaling and find in most cases they nicely match numerical results. On the other hand, we simulate a broad class of practical quantum circuits and find that white noise is in certain cases a bad approximation posing significant limitations on the performance of some of the simpler error mitigation schemes. On a positive note, we find in all cases that the commutator norm is sufficiently small guaranteeing a very good performance of purification-based error mitigation. Lastly, we identify techniques that may decrease both metrics, such as increasing the dimensionality of the dynamical Lie algebra by gate insertions or randomised compiling.
\end{abstract}

\maketitle

\section{Introduction \label{sec:introduction}}
Current generations of quantum hardware can already significantly outperform classical computers
in random sampling tasks~\cite{arute2019quantum,tillmann2013experimental} and hopefully future hardware developments will
enable powerful applications in quantum machine learning~\cite{biamonte2017quantum}, fundamental physics~\cite{jafferis2022traversable,latticeschwinger}
and in developing novel drugs and materials~\cite{cao2019quantum, mcardle2020quantum,bauer2020quantum, Motta2022}.
The scale and precision of the technology today is, however, still below what is required for fully fault-tolerant quantum computation:
Due to noise accumulation in the noisy intermediate-scale quantum (NISQ) era~\cite{preskill2018quantum}, one is thus limited to only
shallow-depth quantum circuits which led to the development of a broad range of hybrid quantum-classical protocols
and quantum machine learning algorithms~\cite{cerezo2021variational,endo2021hybrid,bharti2021noisy}.

The aim in this paradigm is to prevent excessive error buildup via a parameterised, shallow-depth quantum circuit and then perform a series of repeated
measurements in order to extract expected values. These expected values are then post processed on a classical computer in order
to update the parameters of the circuits, e.g., as part of a training procedure.
A major challenge is the potential need for an excessive number of circuit
repetitions which, however, can be significantly suppressed by the use of advanced training algorithms~\cite{cerezo2022challenges,PRXQuantum.2.030324,PhysRevResearch.4.023017}
or via classical-shadows-based protocols~\cite{huang2020predicting,PhysRevX.12.041022,chan2022algorithmic}.
As such, the primary limitation of near-term applications is the damaging effect of gate noise 
on the estimated expected values which can only be reduced by advanced
error mitigation techniques~\cite{cai2022quantum,endo2021hybrid}.

Error mitigation comprises a broad collection of diverse techniques that generally aim to estimate precise expected values by suppressing the
effect of hardware imperfections~\cite{cai2022quantum,endo2021hybrid}. Due to the diversity of techniques and due to the
significant differences in the range of applicability, the need for performance metrics was recently emphasised~\cite{cai2022quantum}.
This motivates the present work to characterise noise in typical practical circuits, e.g., in quantum simulation or in quantum machine learning,
and define two key metrics that determine the performance of a broad class of error mitigation techniques:
(a) eigenvalue uniformity as a closeness to global depolarising (white) noise and (b) norm of the commutator between the ideal and noisy quantum states.
While (b) determines the performance of purification based error mitigation techniques~\cite{PhysRevX.11.031057,PhysRevX.11.041036},
(a) implies a good performance of all error mitigation techniques.

Our primary motivation is that gate errors in complex quantum circuits are scrambled into global white noise~\cite{arute2019quantum,dalzell2021random}.
This property has been proved for random circuits by establishing exponentially decreasing error bounds;
surprisingly, in our numerical simulations we find that in many practical scenarios
the same bounds apply relatively well. In particular, we find that both our metrics, (a) the distance from global-depolarising noise
and (b) the commutator norm, are approximated as
\begin{equation}\label{eq:nu_alpha}
	f(\nu) = \alpha \frac{e^{-\xi} \xi  }{(1-e^{-\xi})  \sqrt \nu  } = \frac{\alpha}{\sqrt \nu } + O(\xi),
\end{equation}
where $\nu$ is the number of gates in the quantum circuit, $\xi$ is the number of expected errors in the
entire circuit and $\alpha$ is a constant.
As such, if one keeps the error rate small $\xi \ll 1$ but increases the number of gates in a circuit
then both (a) and (b) are expected to decrease.
This is a highly desirable property in practice, e.g., 
white noise does not introduce a bias to the expected-value measurement but only a trivial, global
scaling as we detail in the rest of this introduction.

In the present work we simulate a broad range of quantum circuits often used in practice
and identify scenarios where this approximation holds well, by means of gate parameters and circuit structures are sufficiently random.
We also identify strategies that improve scrambling local gate noise
into global white noise, such as inserting additional gates into a circuit
to increase the dimensionality of its Lie algebra~\cite{larocca2022diagnosing}.
In most cases, however, we conclude that
white noise is not necessarily a good approximation due to the large prefactor $\alpha$ in \cref{eq:nu_alpha}. Thus the performance of some error mitigation techniques that rely on a global-depolarising noise assumption is limited. 
On the other hand, we find that in all cases the commutator norm, our other key metric, is smaller
by at least 1-2 orders of magnitude guaranteeing a very good performance of purification-based error mitigation techniques.

Our work is structured as follows. In the rest of this introduction we briefly review
global depolarising noise and how it can be exploited in error mitigation, and then 
briefly review purification-based error mitigation techniques and their performance. 
In \cref{sec:backg} we introduce theoretical notions and finally in \cref{sec:numerics}
we present our simulation results.

\subsection{Global depolarisation and error mitigation\label{sec:intro_glob_dep}}
In the NISQ-era, we don't have comprehensive solutions to error correction, which has led the field to develop error mitigation techniques. These techniques aim to extract expected values $\langle O \rangle_{ideal} := \tr[O \rhoid]$
of observables---typically some Hamiltonian of interest---with respect to an ideal noiseless quantum state $\rhoid$.

A very simple error model, the global depolarising noise channel,
has been very often considered as a relatively good approximation to complex quantum circuits. For qubit states, the channel mixes the ideal, noise-free state with the maximally mixed state 
$\mathrm{Id}/d$ of dimension $d=2^N$ as
\begin{equation}\label{eq:rhown}
	\rhown := \eta \rhoid + (1-\eta) \mathrm{Id}/d.
\end{equation}
Here $\eta \approx F$ is a probability that approximates the fidelity
as $F = \eta + (1-\eta)/d$.
The white noise channel has been commonly used in the literature for modelling errors in near-term quantum computers~\cite{PhysRevE.104.035309}
and, in particular,
it has been shown to be a very good approximation to noise in random circuits~\cite{arute2019quantum,dalzell2021random}.
White noise is extremely convenient as it lets the user extract, after rescaling by $\eta$, the ideal expected value
of any traceless Hermitian observable $O$ via
\begin{equation}	\label{eq:expected_values}
	 \langle O \rangle_{ideal} = \tr[O \rhown] / \eta.
\end{equation}
Of course, for small fidelities $\eta \ll 1$
the expected value $\tr[O \rhown]$ requires a significantly increased
sampling to suppress shot noise.
In this model, the scaling factor $\eta$ is a global property and can be estimated experimentally,
e.g., via randomised measurements~\cite{PhysRevE.104.035309}, via extrapolation~\cite{endo2018practical}
or via learning-based techniques \cite{PRXQuantum.2.040330}.

Global depolarisation, however, may not be sufficiently accurate model to
capture more subtle effects of gate noise
and thus rescaling an experimentally estimated expected value yields 
a biased estimate of the ideal one
as $\langle O \rangle_{bias} := \tr[O \rho]/\eta - \langle O \rangle_{ideal} $.
The bias here $\langle O \rangle_{bias}$ is not a global
property, i.e., it is specific to each observable, and requires the use of
more advanced error mitigation techniques to suppress.

Intuitively, one expects the bias is small for quantum
states that are well approximated by 
a global depolarising model as $\rho \approx \rhown$ and, indeed, we find a general upper bound
in terms of the trace distance as
\begin{equation}\label{eq:trace_dist_bound}
	|\langle O \rangle_{bias}|
	=
	\frac{ | \tr[O \rho] - \tr[O \rho_{wn}] | }{\eta}
	\leq \frac{ \lVert O \rVert_\infty  \lVert \rho - \rho_{wn} \rVert_1  }{\eta}.
\end{equation}
Here $\lVert O \rVert_\infty$ is the operator norm as the absolute largest eigenvalue of the traceless $O$, 
refer to ref.~\cite{koczor2021dominant} for a proof. As such, a small trace distance guarantees a small
bias and thus indirectly determines the performance of all error mitigation techniques --  and further
protocols~\cite{PhysRevA.106.062416,chan2022algorithmic}.

In this work, we characterise how close noisy quantum states $\rho$ in practical applications approach
white noise states $\rho_{wn}$ and consider various types of variational quantum circuits that are typical
for NISQ applications.
When the above trace distance is small then it guarantees a small bias in expected values
which allows us to nearly trivially mitigate the effect of gate noise, i.e., via a simple 
rescaling.

\subsection{Purification-based error mitigation and the commutator norm}
Another core metric we will consider is the commutator norm between the ideal and noisy quantum states as
$\mathcal{E}_C := \lVert [ \rhoid, \rho ] \rVert_1$, which determines the performance
of purification based error mitigation techniques~\cite{koczor2021dominant} -- a small commutator norm has significant practical implications as it guarantees that one can accurately determine expected values
using the ESD/VD~\cite{PhysRevX.11.031057,PhysRevX.11.041036} error mitigation techniques. In particular, independently preparing $n$ copies of the
noisy quantum state and applying a derangement circuit to entangle the copies, allows one to estimate the
expected value
\begin{equation*}
 \frac{\tr[\rho^n O]}{\tr[\rho^n]} =	\langle O \rangle_{ideal} 	+ \mathcal{E}_{ESD}.
\end{equation*}
The approach is very NISQ-friendly~\cite{o2022purification,PhysRevApplied.18.044064}
and its approximation error $\mathcal{E}_{ESD}$ approaches in exponential
order a noise floor as we increase the number of copies $n$~\cite{PhysRevX.11.031057};
This noise floor is determined generally by the commutator norm $\mathcal{E}_C$ but in the most typical
applications of preparing eigenstates, the noise floor is quadratically smaller as $\mathcal{E}_C^2$~\cite{koczor2021dominant}.

Note that this commutator can vanish even if the quantum state is very far from a white noise state,
in fact it generally vanishes when $\rhoid$ approximates an eigenvector of $\rho$.
When a state is close to the white noise approximation then a small commutator norm is guaranteed, 
however, we demonstrate that the latter is a
much less stringent condition and a much better approximation in practice than the former:
in all instances we find that the commutator norm is significantly smaller than the trace distance from white noise states.

\section{Theoretical Background\label{sec:backg}}

In this section we introduce the main theoretical notations and
recapitulate the most relevant results from the literature.

\subsection{General properties of noisy quantum states\label{sec:backg_noisy_states}}

Recall that any quantum state of dimension $d$ can be represented via its density matrix $\rho$ that
generally admits the spectral decomposition as
\begin{equation}\label{eq:spectral_decomp}
	\rho = \sum_{k=1}^{d} \lambda_k \ketbra{\psi_k},
\end{equation}
where we focus on $N$-qubit systems of dimension $d=2^N$.
Here $\lambda_k$ are non-negative eigenvalues and $\ket{\psi_k}$ are eigenvectors.
Since $\sum_i \lambda_i = 1$, the spectrum $\underline{\lambda}$ is also interpreted as a probability distribution.

If $\rho$ is prepared by a perfect, noise-free unitary circuit, 
only one eigenvalue is different from zero and the corresponding eigenvector is the ideal quantum state as $\ket{\psi_{id}}$.
In contrast, an imperfect circuit
prepares a $\rho$ that has more than one non-zero eigenvalues and is thus a probabilistic mixture
 of the pure quantum states $\ket{\psi_k}$, e.g., due to interactions with a surrounding environment.
In fact, noisy quantum circuits that we typically encounter in practice produce quite
particular structure of the eigenvalue distribution: one dominant component that approximates the ideal
quantum state $\ket{\psi_1} \approx \ket{\psi_{id}}$ mixed with an exponentially growing number of
``error'' eigenvectors that have small eigenvalues. White noise is the limiting case
where non-dominant eigenvalues are exponentially small $\propto 1/d$ and $\ket{\psi_1} \approx \ket{\psi_{id}}$.

The quality of the noisy quantum state is then defined by the probability of the ideal quantum state
as the fidelity $F:= \expval{\rho}{\psi_{id}}$; We show in \cref{app:proof_eigval_appr}
that for any quantum state it approaches the dominant eigenvalue $\lambda_1$ as
\begin{equation}\label{eq:eigval_appr}
	\lambda_1 = F + O(\mathcal{E}_C),
\end{equation}
where we compute the error term analytically 
in terms of the commutator norm $\mathcal{E}_C = \lVert [ \rhoid, \rho ] \rVert_1$ from \cref{sec:intro_glob_dep}. This property is actually completely general and applies to any density matrix.

\subsection{Practically motivated noise models\label{sec:pract_quantum_sates}}
Most typical noise models used in practice, such as local depolarising or dephasing noise,
admit the following probabilistic interpretation: 
a noisy gate operation $\Phi(\rho)$ can be interpreted as a mixture of the
noise-free operation $U$ that happens with probability $1-\epsilon$
and an error contribution as
\begin{equation}\label{eq:noise_model}
	\Phi_k(\rho) = (1-\epsilon) U_k \rho U_k^\dagger + \epsilon \Phi_{err} (U_k \rho U_k^\dagger).
\end{equation}
Here $U_k$ is the $k^{th}$ ideal quantum gate and
the completely positive trace-preserving (CPTP) map $\Phi_{err}$ happens with probability $\epsilon$ and accounts for all
error events during the execution of a gate.
A quantum circuit is then a composition of a series of $\nu$ such quantum gates
which prepares the convex combination as
\begin{equation} \label{eq:decomposition}
	\rho = \eta \rhoid + (1-\eta) \rhoerr.
\end{equation}
Here $\rhoid := \ketbra{\psi_{id}}$ is the ideal noise-free state, $\rhoerr$ is an error density matrix
and $\eta = (1-\epsilon)^\nu$ is the probability that none of the gates have undergone errors. 
This probability actually~\cite{koczor2021dominant,dalzell2021random}
approximates the fidelity $F:= \expval{\rho}{\psi_{id}}$ given the noise model in \cref{eq:noise_model} as
\begin{equation}\label{eq:fid_approx}
	F =  (1-\epsilon)^\nu  + \mathcal{E}_{F} =  e^{-\xi} + \mathcal{E}_{F} + O(\epsilon^2 /\nu).
\end{equation}
Here we approximate $(1-\xi/\nu)^\nu = e^{-\xi} + O(\epsilon^2 /\nu)$
for small $\epsilon$ and large $\nu$ where $\xi := \epsilon \nu$ is the circuit error rate 
as the expected number of errors in a circuit. 
In practice the approximation error $\mathcal{E}_{F}=\expval{\rhoerr}{\psi_{id}}$
is typically small and in the limiting case of white noise it decreases exponentially as
$\mathcal{E}_{F} = 1/d$ due to $\rhoerr = \mathrm{Id}/d$.

Assuming sufficiently deep, complex circuits, ref.~\cite{koczor2021dominant} obtained an approximate
bound for the
commutator between the ideal and noisy quantum states as
\begin{equation} \label{eq:commutator_norm}
	\lVert [ \rhoid, \rho ] \rVert_1 	 \lessapprox \mathrm{const} \times e^{ -\xi  } \xi /\sqrt{\nu}.
\end{equation}
This bound confirms that as we increase the number of quantum gates $\nu$ in a circuit
but keeping the circuit error rate $\xi$ constant,
the commutator norm decreases as $\propto 1/\sqrt{\nu}$~\cite{koczor2021dominant}.
Furthermore, this function
closely resembles to \cref{eq:nu_alpha} which is a central aim of this work to explore.

\begin{figure*}[tb]
	\centering
	\includegraphics[width=0.8\textwidth]{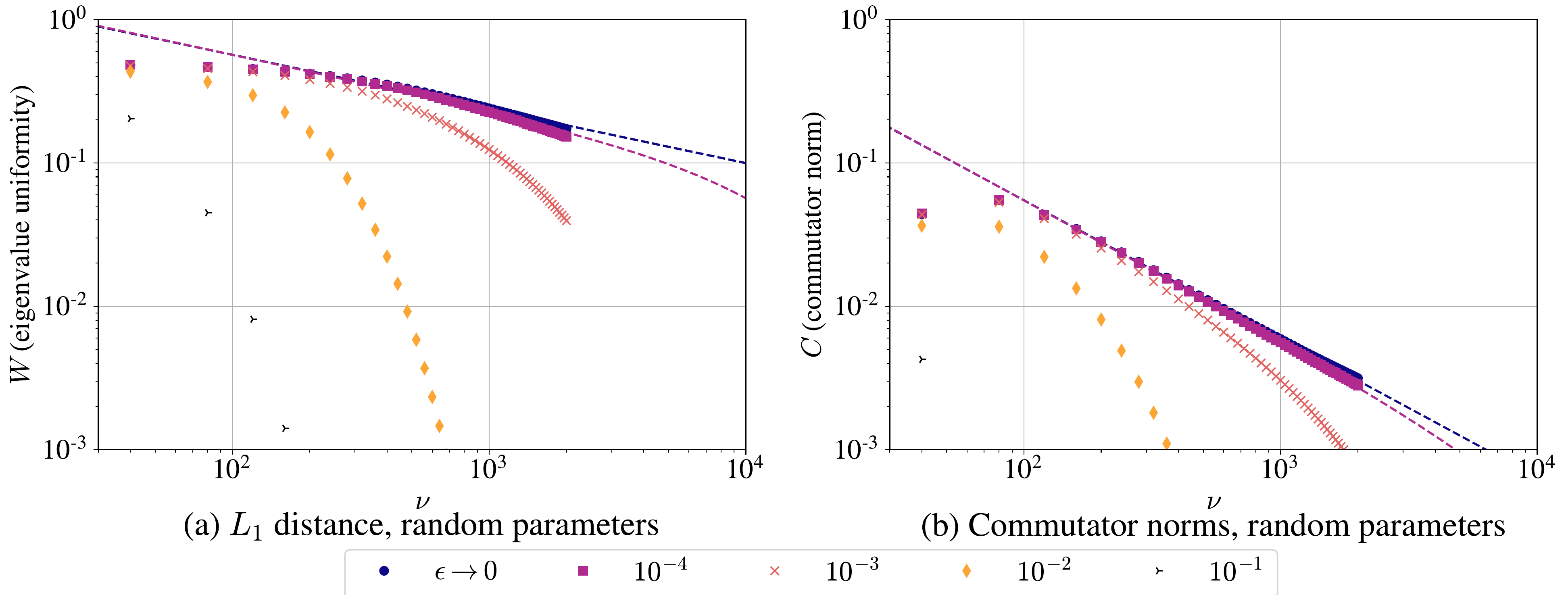}
	\caption{
		Simulating families of 10-qubit Strong entangling layer (SEL) ansatz circuits \cite{schuld2020circuit} at random gate parameters for an increasing number $\nu$ of gates and
		per-gate depolarising error rates $\epsilon$.
		(a) the uniformity measure 	$W(\nu)$ of the error eigenvalues of the density matrix from \cref{eq:prob_dist_for_plots}
		closely match the theoretical model (dashed lines) for random circuits and confirm that increasing the number of gates in random
		circuits scrambles local noise into global white noise.  
		(b) the commutator norm $C(\nu)$ from \cref{eq:commnorm} is significantly smaller in absolute value and
		decreases with a larger polynomial degree (steeper slope of the dashed lines)
		than the uniformity measure -- this suggests that the dominant eigenvector of the density
		matrix $\rho$ approximately commutes with $\rho$ even when noise is not well described by white noise.
		The $\epsilon \rightarrow 0$ simulations were approximated using $\epsilon=10^{-8}$ ($\epsilon=10^{-7}$)
		when calculating $W$ ($C$).
		\label{fig:SEL-ansatz}
	}
\end{figure*}

\subsection{White noise in random circuits~\label{sec:random_intro}}
Random circuits have enabled quantum supremacy experiments using noisy quantum computers for two primary reasons:
(a) the outputs of these circuits are hard to simulate classically  and (b)
they render local noise into global white noise~\cite{arute2019quantum}, hence
introducing only a trivial bias to the ideal probability distribution similarly as in~\cref{sec:intro_glob_dep}.

Ref~\cite{dalzell2021random} considered random circuits consisting of $s$ two-qubit gates,
each of which undergoes two single qubit (depolarising) errors each with probability $\tilde{\epsilon}$
(assuming single-qubit gates are noiseless).
We can relate this to our model by identifying the local noise
after each two-qubit gate with the error event in \cref{eq:noise_model}
via the probability  $\epsilon = 1- (1-\tilde{\epsilon})^2 = 2 \tilde{\epsilon} -  \tilde{\epsilon}^2$.
Ref~\cite{dalzell2021random} then established the fidelity $\tilde{F}$ of the quantum state
which one obtains from a noisy cross-entropy score as
\begin{equation*}
	\tilde{F} = e^{- 2 s \tilde{\epsilon} \pm O(s \tilde{\epsilon}^2)} = e^{- \xi \pm O(\epsilon \xi) }.
\end{equation*}
This coincides with our approximation from \cref{eq:fid_approx}
up to an additive error in the exponent which, however, diminishes
for low gate error rates. In the following we will thus assume $F \equiv \tilde{F}$.

Measuring these noisy states in a the standard measurement basis $\{|j \rangle\}_{j=1}^{d}$
produces a noisy probability distribution $\tilde{p}_{noisy} (j) = \langle j | \rho |j \rangle$. 
Ref.~\cite{dalzell2021random} established that this 
probability distribution rapidly approaches the white noise approximation $\tilde{p}_{wn} = F p_{id}  + (1-F) p_{unif}$.
In particular, the  total variation distance (via the $l_1$ norm $\lVert x \rVert_1 = \sum_i |x_i|$) between the two probability
distributions  is upper bounded as
\begin{equation} \label{eq:theoretical_bound}
	\tfrac{1}{2}   \lVert \tilde{p}_{noisy} -  \tilde{p}_{wn} \rVert_1 
	\leq
	O(F \epsilon \sqrt{\nu})
	=
	O(e^{-\xi} \xi/\sqrt{\nu}).
\end{equation}
This expression is formally identical to the bound
on the commutator norm in \cref{eq:commutator_norm}; Indeed if the noise
in the quantum state approaches a white noise approximation,
it implies that 
the commutator norm must also vanish in the same order.

On the other hand, the reverse is not necessarily true as \cref{eq:theoretical_bound}
is a stronger condition than \cref{eq:commutator_norm}
as the latter only guarantees that
the dominant eigenvector approaches $\ket{\psi_1} \approx \ket{\psi_{id}}$
but does not imply anything about the eigenvalue distribution of $\rho$ or $\rhoerr$.

\section{Numerical simulations\label{sec:numerics}}

\subsection{Target metrics\label{sec:description}}

In the NISQ-era comprehensive error correction will not be feasible and thus hope is primarily based on variational quantum algorithms~\cite{cerezo2021variational,endo2021hybrid,bharti2021noisy,koczor2020variational,foldager2022noise}.
In this paradigm a shallow, parametrised quantum circuit is used
to prepare a parametrised quantum state that aims to approximate the solution to a given problem, typically the ground state of a problem Hamiltonian. Due to its shallow depth the ansatz circuit is believed to be error robust and its tractable parametrisation allows to explore the Hilbert space near the solution. On the other hand, such circuits
are structurally quite different than random quantum circuits and it was already raised in ref.~\cite{dalzell2021random}
whether error bounds on the white noise approximation extend to these shallow quantum circuits.

We simulate such quantum circuits under the effect of local depolarising noise -- while note
that a broad class of local coherent and incoherent error models can effectively be transformed into
local depolarising noise using, e.g., twirling techniques or randomised compiling~\cite{PhysRevA.78.012347, PhysRevA.85.042311, cai2019constructing, cai2020mitigating}.
We analyse the resulting noisy density matrix $\rho$
by calculating the following two quantities.
First, we quantify `closeness' to a white noise state from \cref{eq:rhown}
by computing uniformity measure $W$ as the $l_1$-distance between the uniform distribution and the
non-dominant eigenvalues of the output state as
\begin{equation}\label{eq:prob_dist_for_plots}
W:=	\frac{1}{2} \lVert p_{err}  -p_{unif} \rVert_1 =  \frac{1}{2} \sum_{k=2}^d | \frac{ \lambda_k  } {1-\lambda_1} - \frac{1}{d-1}|,
\end{equation}
which only depends on spectral properties of the quantum state and can thus be computed straightforwardly.
We show in \cref{statement:uniformity} that
$W$ is proportional
to the trace distance from a white noise quantum state as
\begin{equation}
	\lVert \rho - \rhown \rVert_1
	=
	(1-\lambda_1) W
	+\mathcal{E}_w,
\end{equation}
uo to a bounded error $\mathcal{E}_w$.
The uniformity measure $W$ thus determines the bias in estimating any traceless expected value
as discussed in \cref{sec:intro_glob_dep}.

Second, we calculate the commutator norm $\mathcal{E}_C$
from \cref{sec:intro_glob_dep} relative to $1-\lambda_1$ as
\begin{equation}\label{eq:commnorm}
	C:= \frac{
		\lVert [ \rhoid, \rho ] \rVert_1 }
	{ 1-\lambda_1 }
	=
	\lVert [ \rhoid, \rhoerr ] \rVert_1 + \mathcal{O}(\mathcal{E}_q),
\end{equation}
which we relate to the commutator norm between the ``error part'' of the state $\rho_{err}$ 
and the ideal quantum state $\rhoid$ in \cref{lemma:commnorm}.
In the following, we will refer to $C$ as the commutator norm
-- and recall that it determines the ultimate performance of purification-based 
error mitigation as discussed in \cref{sec:intro_glob_dep}.

\begin{figure*}[tb]
	\centering
	\includegraphics[width=0.7\textwidth]{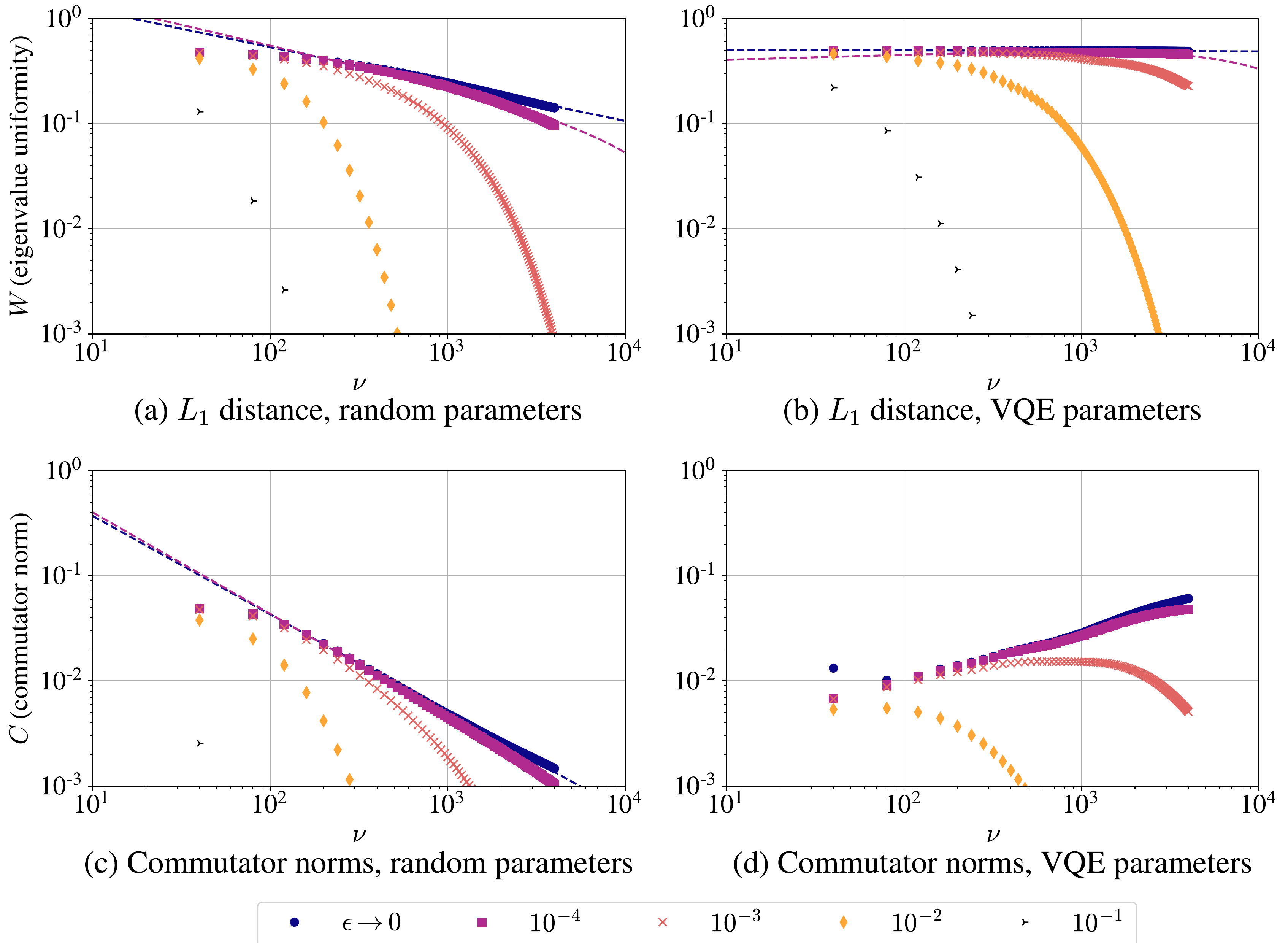}
	\caption{ \textbf{XXX Hamiltonian:}
		same simulations as in \cref{fig:SEL-ansatz} but using 10-qubit HVA quantum circuits constructed for the XXX spin problem Hamiltonian.
		(a, c) at randomly chosen circuit parameters of the HVA we find the same conclusions as for random circuits in \cref{fig:SEL-ansatz}.
		(b) when the HVA circuit approximates the ground state of the Hamiltonian (VQE parameters) we find the
		noise in the circuit is not approximated well by white noise, i.e., the uniformity measure  $W(\nu)$ is large and does not decrease as we increase $\nu$.
		(d) On the other hand, the commutator norm $C(\nu)$ is significantly smaller than $W(\nu)$ confirming that
		the the ideal quantum state approximately commutes with the noisy one.		
		The $\epsilon \rightarrow 0 $ simulations were approximated using $\epsilon=10^{-8}$ ($\epsilon=10^{-7}$)
		when calculating $W$ ($C$).
		\label{fig:xxx}
	}
\end{figure*}

\subsection{Random states via Strong Entangling ans{\"a}tze~\label{sec:var_circ_random_states}}
We first consider a Strong Entangling ansatz (SEA):
it is built of alternating layers
of parametrised single-qubit rotations followed by a series of nearest-neighbour
CNOT gates as illustrated in~\cref{fig:ansatz} -- and assume a local depolarising noise with probability $\epsilon$.
We simulate random quantum circuits
by randomly generating parameters $|\theta_k| \leq 2\pi$
-- note that these circuits are not necessarily Haar-random distributed and thus results in \cref{sec:random_intro} do not necessarily apply.

We simulate $10$-qubit circuits and in \cref{fig:SEL-ansatz} (a) we
plot the eigenvalue uniformity $W(\nu)$ while in 
\cref{fig:SEL-ansatz} (b) we plot the commutator norm $C(\nu)$
for an increasing number $\nu$ of quantum gates -- all datapoints are averages over ten random seeds.
These results surprisingly well recover the expected 
behaviour of random quantum circuits 
as for small error rates $\epsilon \rightarrow 0$ both quantities $W(\nu)$ and  $C(\nu)$
can be approximated by the function from \cref{eq:nu_alpha} as we now discuss.

In \cref{sec:random_intro} we stated bounds 
of ref.~\cite{dalzell2021random}
on the distance between $\tilde{p}_{noisy}$ and $\tilde{p}_{wn}$.
Based on the assumption that these bounds also apply to the probability
distributions $p_{noisy} =  \bra{\psi_k} \rho \ket{\psi_k} $
and $p_{wn} := \bra{\psi_k} \rhown \ket{\psi_k}$
we derive in \cref{statement:upperbound}
the approximate bound on the eigenvalue uniformity as
\begin{equation*}
	W = O\left( \frac{ e^{-\xi} \xi/\sqrt{\nu}  } { 1 - e^{-\xi}}  \right).
\end{equation*}
Furthermore, by combining \cref{eq:commnorm} and the bound in \cref{eq:commutator_norm}
we find that the commutator norm $C$ is similarly bounded by the same function.
On the other, \cref{fig:SEL-ansatz} (b) suggests that the commutator norm
decreases with a larger polynomial degree and thus
we approximate both $W(\nu)$ and $C(\nu)$ using the function
\begin{equation}\label{eq:fitfunction}
f(\nu) =	\alpha \frac{ \xi e^{-\xi}  } {\nu^\beta ( 1 - e^{-\xi})}
=
 \alpha/\nu^\beta + \mathcal{O}(\xi)
\end{equation}
where we fit the two parameters $\alpha$ and $\beta$ to our simulated dataset.
The second equation above is an expansion for small circuit error rates $\xi$ as detailed in \cref{app:expand}.
Indeed, in \cref{fig:SEL-ansatz} (blue circles) for small $\epsilon \rightarrow 0$ we observe a nearly linear function in the
log-log plot in \cref{fig:SEL-ansatz} and thus remarkably well recover the theoretical bounds
with the polynomial power approaching $b \rightarrow 1/2$.

\begin{figure*}[tb]
	\centering
	\includegraphics[width=0.7\textwidth]{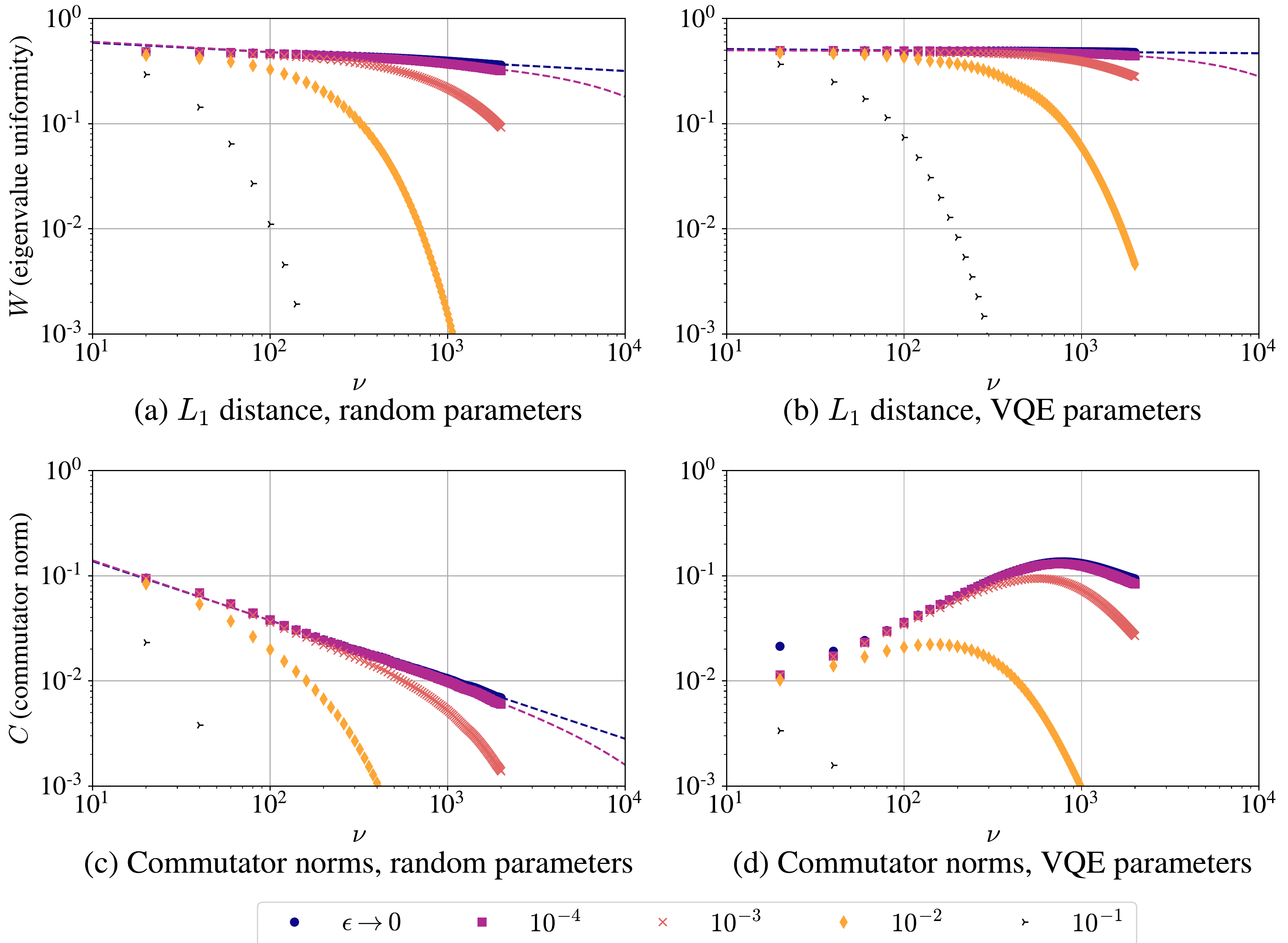}
	\caption{ \textbf{TFI}
		same simulations as in \cref{fig:SEL-ansatz} but using 10-qubit HVA quantum circuits constructed for the TFI spin problem Hamiltonian.
		(a, c) at randomly chosen circuit parameters $W(\nu)$ decreases more slowly, in smaller polynomial order than random circuits -- see text
		and see simulations with added layers of $R_z$ gates in \cref{fig:tfi-hva-scrambled}.
		(b) at the VQE parameters white noise is again not a good approximation, i.e., the uniformity measure  $W(\nu)$ is large and does not decrease as we increase $\nu$.
		(d) the commutator norm $C(\nu)$ is smaller than $W(\nu)$ in absolute value by an order of magnitude.
		The $\epsilon \rightarrow 0 $ simulations were approximated using $\epsilon=10^{-8}$ ($\epsilon=10^{-7}$)
		when calculating $W$ ($C$).
		\label{fig:tfi}
	}
\end{figure*}

Furthermore, comparing \cref{fig:SEL-ansatz} (b, blue circles) and \cref{fig:SEL-ansatz} (a, blue circles)
suggest that the commutator norm has both a significantly smaller absolute value (smaller $\alpha$)
and decreases at a faster polynomial rate (larger beta)  than the uniformity measure.
In fact, the commutator norm is more than two orders of magnitude smaller than the 
uniformity measure which suggests that even when $\rhoerr$ is not approximated well
by a white noise state it, nevertheless, almost commutes with the ideal pure state $\rhoid$. 

We finally	consider how the absolute factor $\alpha$ depends on the number of qubits:
we perform simulations at a small error rate $\epsilon \rightarrow 0$
and fit our model function $\alpha \nu^\beta$ to extract $\alpha(N)$ for an increasing number of qubits.
The results are plotted in \cref{fig:alpha-n} (e) and suggest that the prefactor
$\alpha(N)$ initially grows slowly but then saturates while note that a polylogarithmic depth is sufficient to reach 
anticoncentration~\cite{dalzell2021random}.

\subsection{Variational Hamiltonian Ansatz\label{sec:hva}}
Theoretical results guarantee that the SEL ansatz initialised at random parameters approach
for an increasing depth
unitary 2-designs thereby reproducing properties of random quantum circuits~\cite{mcclean2018barren,cerezo2021cost}.
It is thus not surprising
that the model introduced in \cref{sec:random_intro}
gives a remarkably good agreement between the SEL ansatz (dots on in \cref{fig:SEL-ansatz}) and
genuine random circuits (fits as continuous lines in \cref{fig:SEL-ansatz}).

Here we consider the Hamiltonian Variational Ansatz (HVA)~\cite{wecker2015progress,wiersema2020exploring}
at more practical parameter settings:
The HVA  has the advantage that we can efficiently obtain parameters that increasingly better (as we increase the ansatz depth)
approximate the ground state of a problem Hamiltonian -- we will refer to these as VQE parameters.
We also want to compare this circuit against random circuits and thus also simulate the HVA
such that every gate receives a random parameter as detailed in \cref{app:HVA}.

While the VQE parameter settings capture the relevant behaviour in practice as one approaches a solution,
the random parameters are more relevant, e.g., at the early stages of a VQE parameter optimisation.
Furthermore, as the circuit is entirely composed by Pauli terms in the problem Hamiltonian,
the dimensionality of its dynamical Lie algebra is entirely determined by the problem Hamiltonian in contrast to an
exponentially growing algebra of the SEL ansatz~\cite{larocca2022diagnosing}.

\subsection{Heisenberg XXX spin model \label{sec:var_circ_XXX}}
We first consider a VQE problem of finding the ground state of the 1-dimensional
XXX spin-chain model.
We construct the HVA ansatz from \cref{sec:hva} for this problem Hamiltonian as a sum
$\mathcal{H}_{\text{XXX}} = \mathcal{H}_0 + \mathcal{H}_1$
 as
\begin{equation*}
	\mathcal{H}_0 = \sum_{k=1}^N  \Delta_k Z_k,
	\,
	\mathcal{H}_1 = \sum_{k=1}^N [X_k X_{k+1} {+} Y_k Y_{k+1} {+} Z_k Z_{k+1} ].
\end{equation*}
The Pauli operators $XX$, $YY$ and $ZZ$ determine couplings between nearest neighbour spins in a 1-D
chain and we choose them to be of unit strength. Furthermore, $Z_k$ are local
on-site interactions $|\Delta_k| \leq 1$ that were generate uniformly randomly such that the
Hamiltonian has a non-trivial ground state.

First, we simulate the HVA ansatz for $N=10$ qubits with randomly generated circuit parameters
as $|\theta_k| \leq 2 \pi$ and plot results for an increasing number of quantum gates
in \cref{fig:xxx} (a, c).
We a find similar behaviour for the eigenvalue uniformity $W(\nu)$  as with random SEL circuits in
\cref{fig:SEL-ansatz} (a) 
and obtain a reasonably good fit for $\epsilon \rightarrow 0$ using our model function from \cref{eq:fitfunction}.
The commutator norm in \cref{fig:xxx} (c) is again significantly smaller in magnitude than
the uniformity measure and decreases faster with a higher polynomial order similarly to as with the random
SEL ansatz in \cref{fig:SEL-ansatz} (b) .

Second, in \cref{fig:xxx} (b,d) we simulate the ansatz at the VQE parameters that approximate the ground state.
Since the ansatz parameters become very small as one approaches an adiabatic evolution, it
is not surprising that the output density matrix is not well-approximated by a white noise state:
the uniformity measure is very large in \cref{fig:xxx} (b). 
The commutator norm in \cref{fig:xxx} (d) again, is significantly smaller
than $W(\nu)$ and although it appears to slowly grow with $\nu$,
it appears to decrease for $\nu \rightarrow \infty$. This agrees with observations of ref.~\cite{koczor2021dominant}
that the circuits need not be random for the commutator to be sufficiently small in practice.

Furthermore, in \cref{fig:alpha-n} (a, b) we investigate the dependence on $N$ and find that the prefactor $\alpha$ grows slowly and appears to saturate
for $N\geq 10$.

\begin{figure*}[tb]
	\centering
	\includegraphics[width=0.7\textwidth]{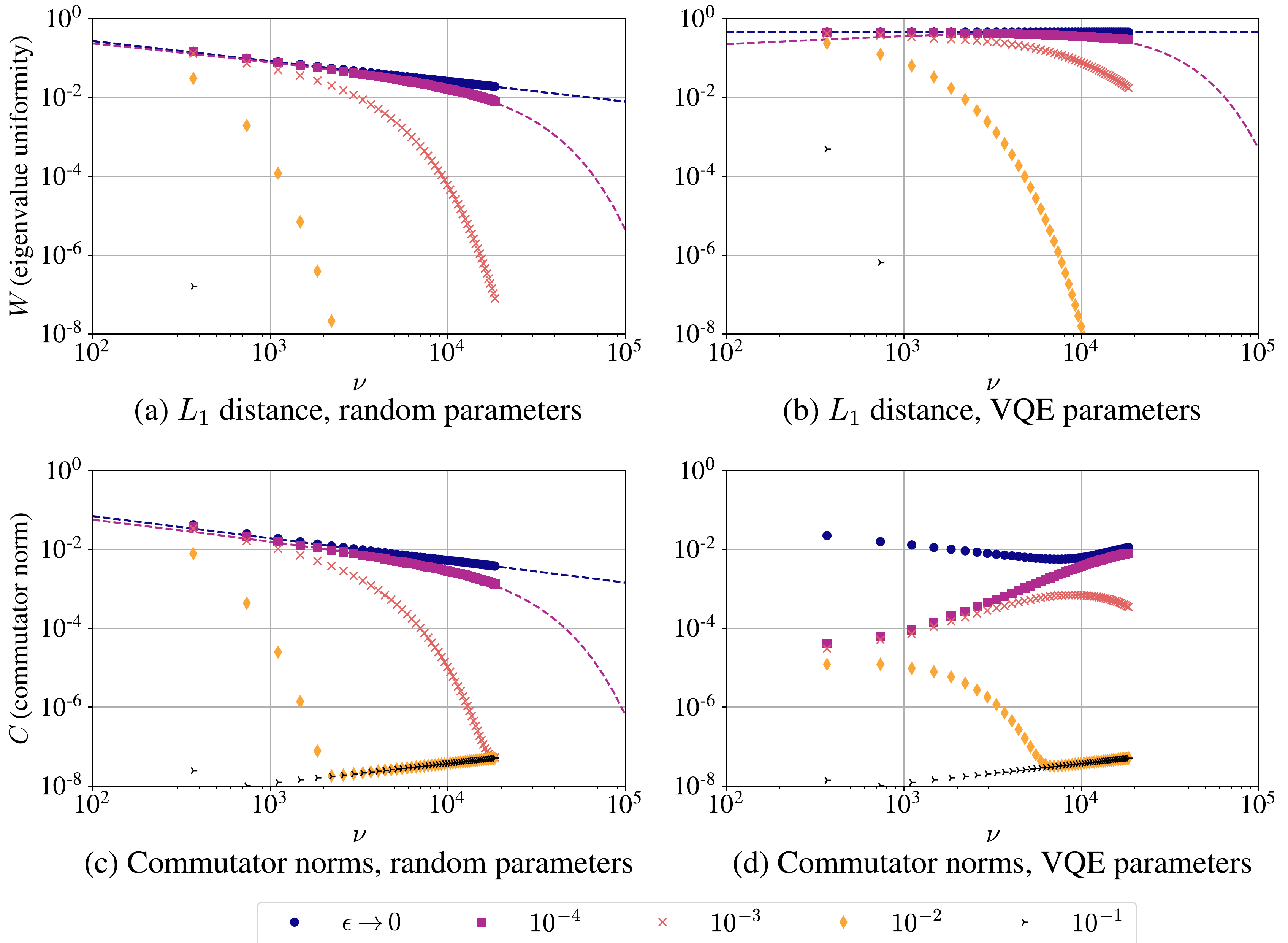}
	\caption{\textbf{LiH} same simulations as in \cref{fig:SEL-ansatz} but using 6-qubit HVA quantum circuits constructed for a LiH molecular Hamiltonian.
		(a, c) at randomly chosen circuit parameters both $W(\nu)$ and $C(\nu)$  decrease as expected for random circuits due our
		randomised compiling strategy~\cite{Campbell_2019,Ouyang2020compilation}.
		(b) at the VQE parameters white noise is an increasingly bad approximation, i.e., the uniformity measure  $W(\nu)$ increases as we increase $\nu$.
		(d) the commutator norm $C(\nu)$ is smaller than $W(\nu)$ in absolute value by 2 orders of magnitude.
		The $\epsilon \rightarrow 0 $ simulations were approximated using $\epsilon=10^{-8}$ ($\epsilon=10^{-7}$)
		when calculating $W$ ($C$).
		\label{fig:lih}
	}
\end{figure*}

\subsection{TFI\label{sec:tfi}}
In the next example we consider the transverse-field Ising (TFI) model $\mathcal{H}_{\text{TFI}} = \mathcal{H}_0 + \mathcal{H}_1$ using
constant on-site interactions $h_i=1$ and 
randomly generated coupling strengths $|J_i| \leq 1$ as 
\begin{equation}
	\mathcal{H}_0 = - \sum_i h_i X_i,
	\quad
	\mathcal{H}_1 = - \sum_i J_{i} Z_i Z_{i+1} \label{eq:tfi}.
\end{equation}
We first simulate the HVA ansatz with random variational
parameters in \cref{fig:tfi} (a, c).
While at small error rates $\epsilon \rightarrow 0$ \cref{fig:tfi} (a, blue) can be fitted well with our
polynomial approximation form \cref{eq:fitfunction},
we observe that the eigenvalue uniformity $W(\nu)$ in \cref{fig:tfi} (a, blue) decreases with a small
polynomial degree.

Indeed, as the HVA ansatz is specific to a particular
Hamiltonian, its dynamical Lie algebra may have a low dimensionality~\cite{larocca2022diagnosing}
resulting in a limited ability to scramble local noise into white noise; this explains why 
in \cref{fig:tfi} (a) the uniformity measure decreases more slowly, i.e., in a smaller polynomial order, 
than random circuits.
For this reason, we additionally simulate in \cref{fig:tfi-hva-scrambled} the TFI-HVA ansatz but with adding 
$R_z$ gates in each layer whose generator is not contained in the problem Hamiltonian.
The increased dimensionality of the dynamic Lie algebra, indeed, improves
scrambling as the white noise approximation is clearly better in \cref{fig:tfi-hva-scrambled}
--
while note that the increased dimensionality may
also lead to exponential inefficiencies in training the circuit~\cite{larocca2022diagnosing}.

In stark contrast to the case of the uniformity measure $W(\nu)$, we find that the commutator norm in
\cref{fig:tfi} (c, blue) decreases substantially for an increasing $\nu$ despite the low dimensionality of the
Lie algebra. This nicely demonstrates that a small commutator norm
is a much more relaxed condition than white noise 
as the latter requires that the noise is fully scrambled in the entirety of the exponentially
large Hilbert space.
Finally, we simulate the TFI circuits at VQE parameters and find qualitatively the same behaviour as in the case of
the XXX problem.

\subsection{Quantum Chemistry: LiH \label{sec:var_circ_lih}}
We consider a 6-qubit Lithium Hydride (LiH) Hamiltonian in the Jordan-Wigner encoding
as a linear combination of non-local Pauli strings 
$P_k \in \{\mathrm{Id}, X, Y, Z \}^{\otimes N}$ as
\begin{equation}\label{eq:lih}
	\mathcal{H}_{LiH} = \sum_{k=1}^{r_h} h_k P_k.
\end{equation}
We construct the HVA ansatz by splitting this Hamiltonian into two parts with $\mathcal{H}_0$ being 
composed of the diagonal Pauli terms in \cref{eq:lih}
while $\mathcal{H}_1$ composed of non-diagonal Pauli strings.

Such chemical Hamiltonians typically have a very large number of terms with $r_h \gg 1$
but a significant fraction only have small weights $h_k$
thus the HVA would have a large number of gates with only very small rotation angles.
For these reasons we construct a more efficient circuit whose basic building blocks are constructed using sparse compilation techniques~\cite{Campbell_2019,Ouyang2020compilation}:
Each single layer in the HVA ansatz consists of gates that correspond to $100$ randomly selected terms of the Hamiltonian
with sampling probabilities $p_k \propto |h_k|$ proportional to the Pauli coefficients.
This approach has the added benefit that it makes the circuit structure random as opposed to the fixed structures in \cref{sec:var_circ_XXX} and in \cref{sec:tfi}.

Results shown in \cref{fig:lih} (a,c) agree with our findings from the previous sections:
at randomly chosen circuit parameters the uniformity measure decreases according to \cref{eq:fitfunction};
the commutator norm similarly decreases but in a higher polynomial order while its absolute value is smaller by at least an order of magnitude.
In contrast, \cref{fig:lih} (b) suggests that the errors are not well approximated by white noise with a large and non-decreasing $W(\nu) \approx 0.5$.
Furthermore, \cref{fig:lih} (b) again confirms that despite white noise is not a good approximation,
the commutator norm is small in absolute value, i.e., $\approx 10^{-3}$ in the practically relevant region. This guarantees a very good
performance of the ESD/VD error mitigation techniques sufficient for nearly all practical purposes.

\section{Discussion} \label{sec:discussion}

Random quantum circuits---instrumental for demonstrating quantum advantage---are known to scramble local gate noise into global white noise
for sufficiently long circuit depths~\cite{arute2019quantum}: general bounds have been proved
on the approximation error which decrease as $\nu^{-1/2}$ as we increase the number $\nu$
of gates in the random circuit~\cite{dalzell2021random}.

In this work we consider shallow-depth, variational quantum circuits that are typical in practical applications of near-term quantum computers
and answer the question: can variational quantum circuits scramble local gate noise into global depolarising noise? 
While the answer to this question is relevant for the fundamental understanding of noise processes in near-term quantum devices,  
it has significant implications in practice: the degree to which local noise is scrambled into white noise determines
the performance of a broad class of error mitigation techniques that are of key importance to achieving value with near-term devices~\cite{cai2022quantum}.
As such, we derive two simple metrics that bound performance guarantees: first, the uniformity
measure $W$ characterises the performance of error mitigation techniques that assume global depolarising (white) noise~\cite{PhysRevE.104.035309}; second,
the norm $C$ of the commutator between the ideal and noisy quantum states determines the performance of purification-based error mitigation techniques~\cite{PhysRevX.11.031057,PhysRevX.11.041036} via bounds of ref.~\cite{koczor2021dominant}.  

We perform a comprehensive set of numerical experiments to simulate typical applications of near-term quantum computers and analyse
characteristics of noise based on the aforementioned two metrics.
In all experiments in which we \emph{randomly initialise parameters of the variational circuits} we semiquantitatively find the same conclusions.
First, both metrics, the eigenvalue uniformity $W$ and the commutator norm $C$ are well described by our polynomial 
approximation from \cref{eq:fitfunction} for small gate error rates. Second, this confirms that, similarly to genuine random circuits,
local errors get scrambled into global white noise with a polynomially decreasing approximation error as we increase the number of gates.
Third, the commutator $C$ decreases at a higher polynomial rate and has a significantly, by 1-2 orders of magnitude, smaller absolute value in the practically relevant region than the eigenvalue uniformity $W$. This confirms that purification based techniques are expected to have a superior 
performance compared to error mitigation techniques that, e.g., assume a global depolarising noise.

We then investigate the practically more relevant case when the ansatz circuits are initialised near the
ground state of a problem Hamiltonian; in all cases we semiquantitatively find the same conclusions.
First, the errors do not get scrambled into white noise and the approximation errors are large thus effectively prohibiting
or at least significantly limiting
the use of error mitigation techniques that assume global depolarising noise.
Second, the commutator norm is quite small in absolute value, i.e., $\approx 10^{-2} -10^{-4}$ in the practically relevant region;
Since the ansatz circuit prepares 
the ground state, the square of the commutator norm determines the performance of ESD/VD
thus for all applications we simulated we expect a very good performance of the ESD/VD approach.
Third, we identify strategies to improve scrambling of local
noise into global white noise as we increase circuit depth: We find that inserting additional gates
to a HVA that is otherwise not contained in the problem Hamiltonian increases the dimensionality of
the dynamic Lie algebra and thus leads to a reduction of both metrics.
We find that applying randomised compiling to these non-random, practical circuits also reduces both metrics.

While purification-based techniques~\cite{PhysRevX.11.031057,PhysRevX.11.041036} have been shown to perform well on specific
examples, the present systematic analysis of circuit noise puts these results into perspective and demonstrates the following:
First, the superior performance of the ESD/VD technique is not necessarily due to randomness in the quantum circuits -- albeit, 
in deep and random circuits its performance is further improved. Second, while some error mitigation techniques perform well
on quantum circuits well-described by white noise~\cite{PhysRevE.104.035309,endo2018practical,PRXQuantum.2.040330},
we identify various practical scenarios where a limited performance is expected.

The present work advances our understanding of the nature of noise in near-term quantum computers and helps making
progress towards achieving value with noisy quantum machines in practical applications.
As such, results of the present work will be instrumental for identifying design principles that lead to robust, error-tolerant
quantum circuits in practical applications.

\vspace{2mm}
\noindent\textbf{Data availability}\\[0mm]
Numerical simulation code is openly available in the repository:
\href{https://github.com/jfold/shallow-circuits-noise}{github.com/jfold/shallow-circuit-noise}.

\section*{Acknowledgements} \label{sec:acknowledgements}
The authors thank Simon Benjamin for his support.
The authors thank Richard Meister for instructions on pyQuEST and Arthur Rattew for
multiple discussions on quantum simulations. All simulations in this work were performed using the simulation tools
QuEST~\cite{Jones2019} and its Python interface pyQuEST~\cite{pyquest}.
B.K. thanks the University of Oxford for
a Glasstone Research Fellowship and Lady Margaret Hall, Oxford for a Research Fellowship.
B.K. derived analytical results and contributed to writing the manuscript. 
J.F. was supported by the William Demant Foundation [grant number 18-4438]. J.F. performed numerical simulations and contributed to writing the manuscript.

\appendix

\begin{figure}[b]
	\begin{quantikz}
		\qw & \gate{R_z} & \gate{R_y} & \gate{R_z} & \ctrl{1} & \qw & \targ{}  & \qw \\
		\qw & \gate{R_z} & \gate{R_y} & \gate{R_z} & \targ{} & \ctrl{1} & \qw  & \qw \\
		\qw & \gate{R_z} & \gate{R_y} & \gate{R_z} & \qw & \targ{}& \ctrl{-2} & \qw \\
	\end{quantikz}
	\caption{A single layer of the Strong Entangling Layers ansatz for three qubits:
		it first applies single-qubit gates $Ry$, $Rz$ and $Ry$
		on all qubits which is then followed by nearest neighbour CNOT gates.
		\label{fig:ansatz}}
\end{figure}
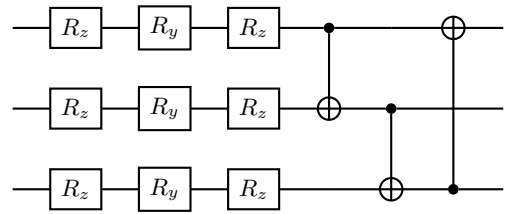
\section{Derivation of \cref{eq:eigval_appr}} \label{app:proof_eigval_appr}
Recall that any quantum state can be transformed into
a non-negative arrowhead matrix following Statement~1 from \cite{koczor2021dominant}
as $\tilde{\rho} = F |\tilde{\psi}_{id}\rangle\langle \tilde{\psi}_{id}| + D+C$ with
\begin{equation}
	\tilde{\rho}  = \begin{pmatrix}
		F & C_2 & C_3 & \dots & C_{d}\\
		C_2 & D_2 &  &  &  \\
		C_3 &  & D_3 \\
		\vdots & & & \ddots &\vdots \\
		C_{d} &  &  & \dots & D_{d}
	\end{pmatrix}.
\end{equation}
We obtain the above matrix by applying a suitable unitary transformation $\tilde{\rho} := U \rho U^\dagger$
such that $|\tilde{\psi}_{id}\rangle := U |\psi_{id}\rangle = (1,0, \dots 0)$ 
while $F,C_k,D_k \geq 0$ with $k\in\{2, 3, \dots,d \}$ with $d$ denoting the dimension,
and all other matrix entries are zero. 
Given the above arrowhead representation of a quantum state, one can analytically compute
eigenvalues of the density matrix as roots of the following secular equation~\cite{koczor2021dominant,o1990computing}
\begin{equation} \label{eq_eigenvalue} 
	P(x) = x - F + \sum_{k=2}^d \frac{C_k^2 }{(D_k - x) } =0.
\end{equation}
With this we compute the deviation between dominant eigenvalue $\lambda_1$ and the fidelity as
\begin{align} \nonumber
	\lambda_1 - F =&  \sum_{k=2}^d \frac{C_k^2 }{(\lambda_1 - D_k) }
	\leq
	\max_k (\lambda_1 - D_k)^{-1} \sum_{k=2}^d C_k^2\\
	\leq&  \lVert [\rho_{id}, \rho] \rVert^2 (2\lambda_1 -1)^{-1},
\end{align}
where we have used that $D_k \leq \lambda_1$ and that all summands are non-negative as $D_k, C_k, \lambda_1 \geq 0$,
and in the second inequality we have used the series of matrix norms
$\sum_{k=2} C_k^2 =  \lVert C \rVert_{HS}^2 /2 = \lVert [\rho_{id}, \rho] \rVert^2_\infty$
as established in \cite{koczor2021dominant}.
We have also introduced the abbreviation $\lVert [\rho_{id}, \rho] \rVert$
given all $p$-norms of the matrix $[\rho_{id}, \rho]$
are equivalent up to a constant factor. In particular,
any $p$-norm of the commutator can be computed as $\lVert [\rho_{id}, \rho] \rVert_p = 2^{1/p} \sqrt{\mathrm{Var}[\rho]}$
where we used the quantum mechanical variance $\mathrm{Var}[\rho] := \expval{\rho^2}{\psi_{id}} - F^2$
as established in~\cite{koczor2021dominant}.
Furthermore, in the second inequality in~\cref{eq_eigenvalue}
we have used that $\max_k (\lambda_1 - D_k)^{-1} = (\lambda_1 - D_2)^{-1} \leq (\lambda_1 - \lambda_2)^{-1} \leq (2\lambda_1 -1)^{-1} $
by substituting the general inequality $\lambda_2 \leq (1-\lambda_1)$ due to the fact that $\tr[\rho]=1$.

By denoting the commutator norm as $\mathcal{E}_C$, we can thus finally conclude that
$\lambda_1 - F \in O(\mathcal{E}_C)$
as stated in \cref{eq:eigval_appr}.

\subsection{Trace distance from white noise states}
In this section we evaluate analytically the trace distance of any quantum state $\rho$ from
the corresponding white noise state in \cref{eq:rhown} in terms of a distance between probability distributions.

\begin{statement}\label{statement:wn}
	We can approximate the white noise-state in \cref{eq:rhown} in terms of the dominant eigengvalue $\lambda_1$ and
	the dominant eigenvector $\ket{\psi_1}$ of the quantum state as
	\begin{equation}\label{eq:wn_approx_eigv}
		\rhown = \lambda \ketbra{\psi_1} + (1-\lambda_1) \mathrm{Id}/d + \mathcal{E}_w,
	\end{equation}
	up to an approximation error $\mathcal{E}_w$ that is bounded via \cref{eq:error_wn}.
\end{statement}
\begin{proof}
	We start by approximating the weight $\eta$ in \cref{eq:rhown} as
	$\eta \approx F \approx \lambda_1$ via \cref{eq:fid_approx}
	as well as we approximate the dominant eigenvalue using \cref{eq:eigval_appr} 
	and then collect the approximation errors as
	\begin{equation*}
		\rhown = \lambda \ketbra{\psi_{id}} + (1-\lambda_1) \mathrm{Id}/d + \mathcal{E}_F + \mathcal{E}_C + O(\epsilon^2 /\nu).
	\end{equation*}
	We now use results in \cite{koczor2021dominant} for bounding the distance between
	the ideal and noisy quantum states as
	\begin{align*}
		\lVert \ketbra{\psi_{id}} - \ketbra{\psi_1} \rVert_1 &=
		\sqrt{1- \langle \psi_{id} | \psi_1 \rangle }\\
		&= 1 -  O\left( \frac{ \mathcal{E}_C }{\lambda_1 - \lambda_2} \right),
	\end{align*}
	where $\mathcal{E}_C$ is the commutator norm from \cref{eq:eigval_appr}.
	We thus establish the approximation 
	\begin{equation}
		\rhown = \lambda \ketbra{\psi_1} + (1-\lambda_1) \mathrm{Id}/d + \mathcal{E}_w,
	\end{equation}
	where we collect all approximation errors as
	\begin{equation}\label{eq:error_wn}
		|\mathcal{E}_w| \leq |\mathcal{E}_F| + O(\epsilon^2 /\nu) 
		+O\left[ \mathcal{E}_C (1 + \frac{ 1 }{1 - \lambda_2/\lambda_1} )\right].
	\end{equation}

\end{proof}

\begin{statement}\label{statement:uniformity}
		We define the eigenvalue uniformity as $W:=\tfrac{1}{2} \lVert p_{err}  -p_{unif} \rVert_1$
	via the non-dominant eigenvalues of the density matrix 
	$p_{err}:= (\lambda_2, \lambda_3, \dots, \lambda_d)/(1-\lambda_1)$.
	This metric is related to the trace distance from a white noise state (as in \cref{eq:trace_dist_bound}) as
	\begin{equation}
		\lVert \rho - \rhown \rVert_1
		=
		(1-\lambda_1) W
		+\mathcal{E}_w,
	\end{equation}
	where the approximation error $\mathcal{E}_w$ is stated in \cref{statement:wn}.
\end{statement}
\begin{proof}
	We substitute the approximation of $\rhown$ from \cref{eq:wn_approx_eigv} including the
	error term $\mathcal{E}_w$ and then we
	use the spectral decomposition of $\rho$ to obtain the trace distance as
	\begin{align}\nonumber
		\lVert \rho - \rhown \rVert_1 =& \lVert \sum_{k=2}^{d} \lambda_k \ketbra{\psi_k}   - (1-\lambda_1) \mathrm{Id}/d \rVert_1 + \mathcal{E}_w 
		\\
		=& \frac{1}{2} \sum_{k=2}^{d} |\lambda_k   - \frac{1-\lambda_1}{d} | + \mathcal{E}_w\\
		=& \frac{1-\lambda_1}{2} \lVert p_{err}  -p_{unif} \rVert_1 + \mathcal{E}_w.
	\end{align}
In the second equation we analytically evaluated the trace distance
and thus in the third equation we rewrite the result in terms of $p_{err}$ which is our ``error probability'' distribution as 
	$p_{err} := (\lambda_2, \lambda_3, \dots, \lambda_d)/(1-\lambda_1)$.
\end{proof}

\begin{statement}
Alternatively to \cref{statement:uniformity},
if a quantum state admits the decomposition in \cref{eq:decomposition} 
then we can state the trace distance without approximation as
	\begin{equation}
		\lVert \rho - \rhown \rVert_1 = \frac{(1-\eta)}{2} \lVert p_\mu - p_{unif} \rVert_1.
	\end{equation}
This is directly analogous to the uniformity measure of the non-dominant eigenvalues of $\rho$ in \cref{statement:uniformity}, however, this
expression quantifies the uniformity of the probability distribution  $p_\mu$ which are eigenvalues of the error density matrix $\rhoerr$.
\end{statement}

Let us assume the decomposition in \cref{eq:decomposition}.
We find the following result via a direct calculation as
\begin{align*}
	\lVert \rho - \rhown \rVert_1 &=   (1-\eta) \lVert \rhoerr - \mathrm{Id}/d   \lVert_1\\
	&    =  (1-\eta) \lVert \sum_{k=1}^d \mu_k \ketbra{\phi_k} -  \mathrm{Id}/d \lVert_1\\
	& 	 =  \frac{(1-\eta)}{2} \lVert \sum_{k=1}^d |\mu_k  - 1/d|\\
	& = \frac{(1-\eta)}{2} \lVert p_\mu - p_{unif} \rVert_1
\end{align*}
where we have used the spectral resolution of the error density matrix and
then analytically evaluated the trace distance.
Given $\rhoerr$ is a positive-semidefinite matrix with unit trace, its eigenvalues $\mu_k$
form a probability distribution that we denote as $p_\mu$.

\subsection{Upper bounding the uniformity measure \label{app:err_prob_measurements}}
In this section we upper bound the uniformity measure based on the number of gates and error rates in a quantum circuit.

\begin{statement}\label{statement:upperbound}
	We adopt the bounds of \cite{dalzell2021random} in \cref{eq:theoretical_bound} for the distance between probability
	distributions measured in the standard basis
	$\tfrac{1}{2}   \lVert \tilde{p}_{noisy} -  \tilde{p}_{wn} \rVert_1 $ 
	and assume the same bounds approximately apply to any measurement basis.
	Then, it follows that the uniformity measure from \cref{statement:uniformity} is approximately bounded  by the same bounds as
	\begin{equation*}
	W
	=
	O( \frac{ e^{-\xi} \xi/\sqrt{\nu}  } { 1 - e^{-\xi}}  ) + O( \frac{\mathcal{E}_w }{1-\lambda_1} ),
	\end{equation*}
	where the approximation error $\mathcal{E}_w$ is stated in \cref{statement:wn}.
\end{statement}
\begin{proof}
Let us consider measurements performed in the basis as the eigenvectors
of the density matrix which yield probabilities as the eigenvalues as
$$p_{noisy} = \bra{\psi_k} \rho \ket{\psi_k} = (\lambda_1, \lambda_2 \dots, \lambda_d). $$
Measuring the white noise state in the same basis yields the following
approximation of the probabilities using the error term from 
\cref{eq:wn_approx_eigv} as
\begin{align*}
	p_{wn} &:= \bra{\psi_k} \rhown \ket{\psi_k}\\
	&= (\lambda_1, \frac{1-\lambda_1}{d} \dots, \frac{1-\lambda_1}{d}) + \mathcal{E}_w.
\end{align*}

\begin{figure}[tb]
	\centering
	\includegraphics[width=\columnwidth]{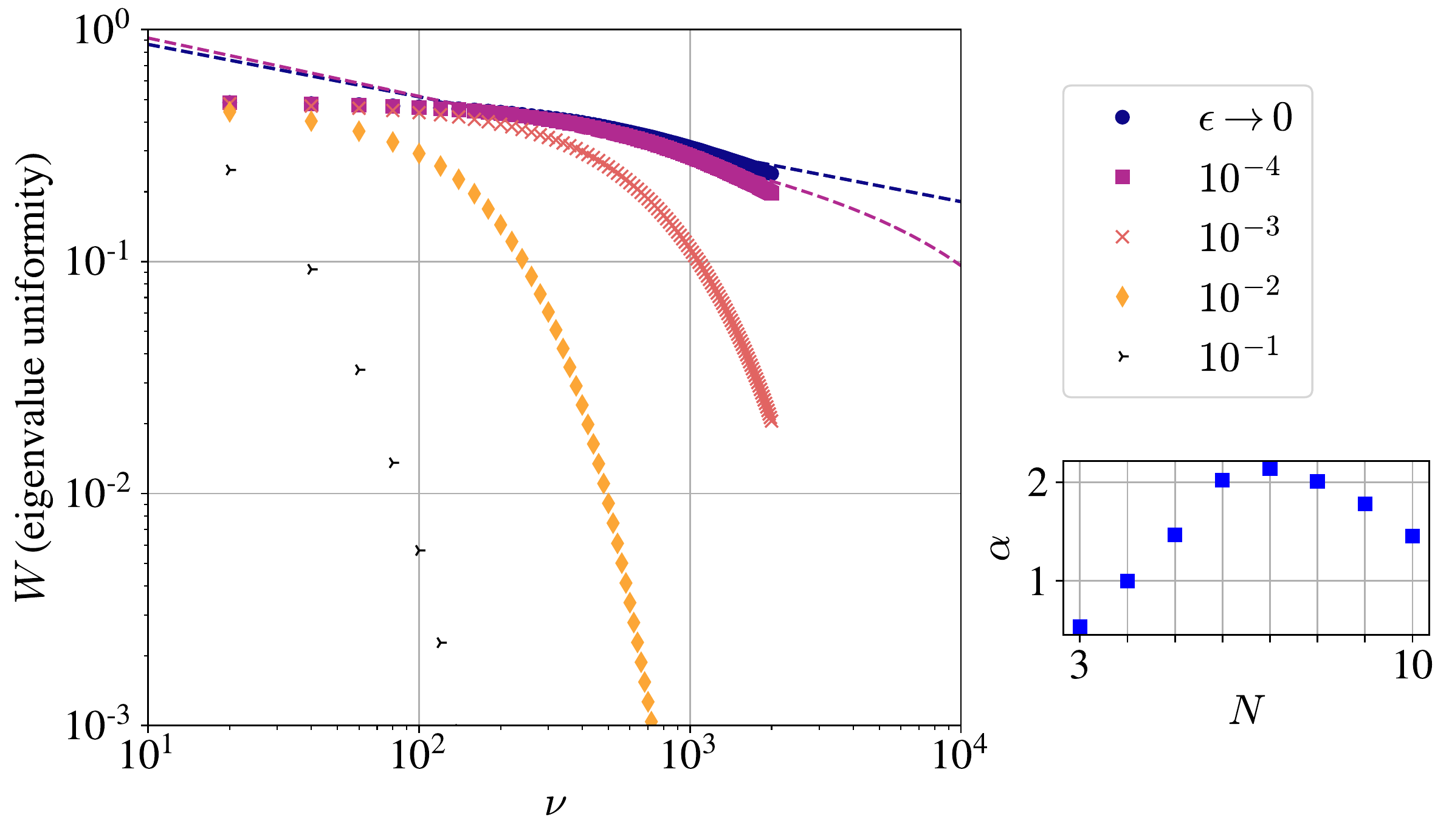}
	\caption{(left) TFI-HVA ansatz: same simulations as in \cref{fig:tfi} (a) but with added parametrised $R_z$ gates after each layer.
		The additional gates increase the dimensionality of the dynamic Lie algebra which leads to a faster 
		scrambling of local gate noise into white noise, e.g.,  the $\epsilon \rightarrow 0$ curve is steeper
		than in \cref{fig:tfi} (a).	See \cref{app:numerics} for more details.
		(right) the dependence on the number of qubits shows a very similar trend as without the $Rz$ gates, i.e., 
		compare to \cref{fig:alpha-n} (c).
	}
	\label{fig:tfi-hva-scrambled}
\end{figure}

The distance of the above two measurement probability distributions is then
\begin{equation*}
	\frac{1}{2} \lVert   p_{noisy} - p_{wn}  \rVert_1
	= 
	(1-\lambda_1)  W + \mathcal{E}_w,
\end{equation*}
where $W = \tfrac{1}{2}\lVert p_{err}  -p_{unif} \rVert_1$
is our  eigenvalue uniformity from \cref{statement:uniformity}.
Under the assumption that the upper bound on the measurement probabilities $\tfrac{1}{2}   \lVert \tilde{p}_{noisy} -  \tilde{p}_{wn} \rVert_1 $ 
from \cref{eq:theoretical_bound} approximately holds for any measurement basis we can bound the eigenvalue uniformity as
\begin{align*} 
		W
		&= \frac{1}{2(1-\lambda_1)} \lVert   p_{noisy} - p_{wn}  \rVert_1 + \frac{ \mathcal{E}_w }{1-\lambda_1}\\
		&\leq
		O\left(\frac{F}{1-\lambda_1} \epsilon \sqrt{\nu}\right) + \frac{ \mathcal{E}_w }{1-\lambda_1}\\
		&=
		O\left( \frac{ e^{-\xi} \xi/\sqrt{\nu}  } { 1 - e^{-\xi}}  \right) + O\left( \frac{\mathcal{E}_w }{1-\lambda_1} \right).
	\end{align*}
In the last equation we introduced the approximation of $F$ from \cref{eq:fid_approx}
as well as the approximate dominant eigenvalue from \cref{eq:eigval_appr}.
\end{proof}

\subsubsection{Expanding the upper bound\label{app:expand}}
We now expand the upper bound from \cref{statement:upperbound}
for small $\xi$ as.
More specifically, we consider the parametrised fit function from  \cref{eq:fitfunction} 
and substitute the Taylor expansion $e^{-\xi}  = 1 - \xi + \xi^2 + \dots$ as
\begin{align*}
	\alpha \frac{e^{-\xi}  \xi/\sqrt{\nu}^\beta }{1-e^{-\xi}}
	&=
	\alpha \frac{e^{-\xi}}{\nu^\beta} 
	\frac{\xi }{
		\xi - \xi^2/2 + \dots
	}\\
	&=
	\alpha \frac{e^{-\xi}}{\nu^\beta} 
	\frac{1 }{
		1 - \xi/2 + \dots
	}\\
&=
\alpha \frac{1}{\nu^\beta} 
\frac{ 1 - \xi +  \dots }{
	1 - \xi/2  + \dots
}\\
	&=   \frac{\alpha}{\nu^\beta  } + O(\xi).
\end{align*}

\begin{figure*}[tb]
	\centering
	\includegraphics[width=0.7\textwidth]{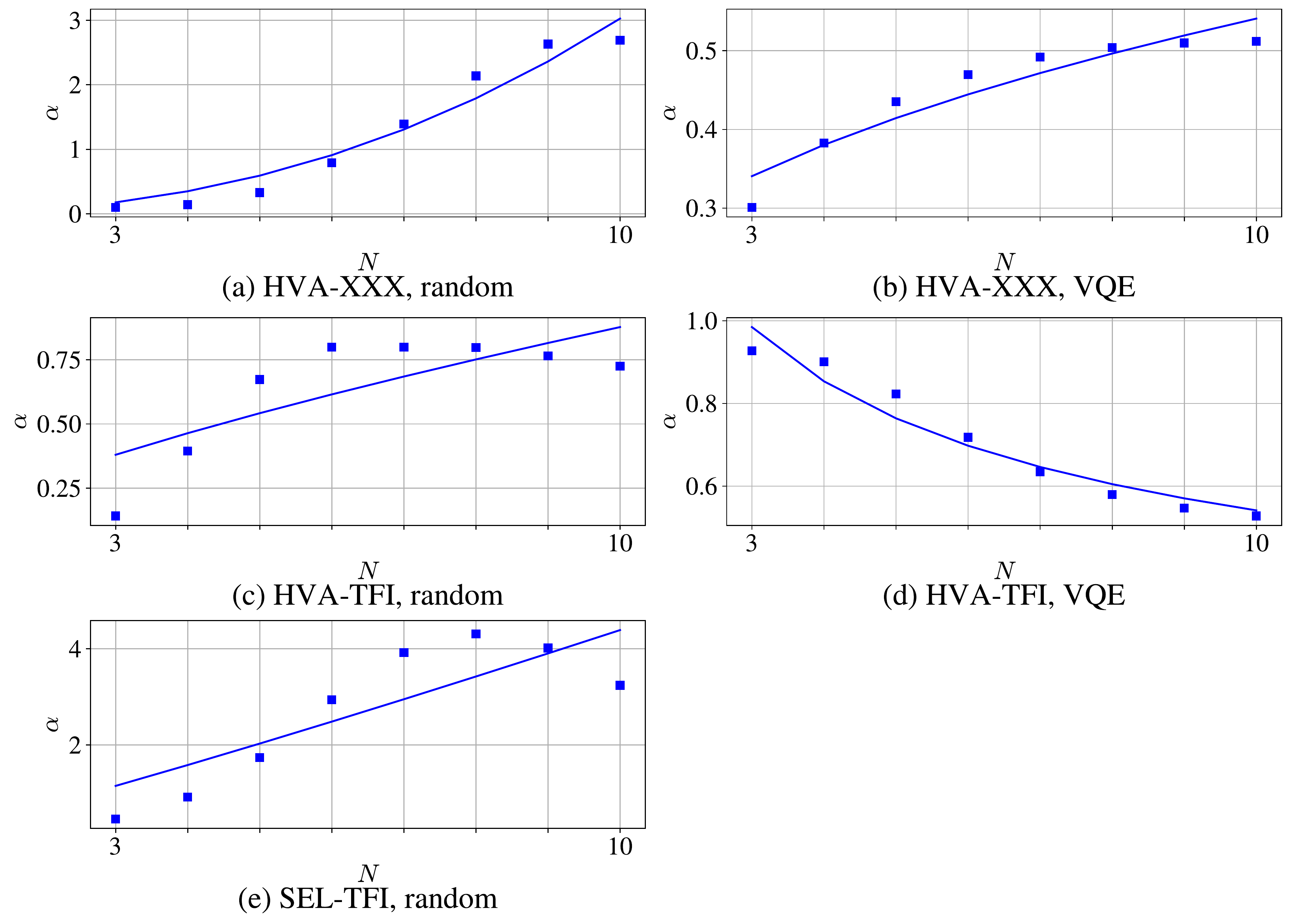}
	\caption{
		Fit parameters $\alpha$ from \cref{eq:fitfunction} for an increasing number of qubits: The circuits in \cref{fig:SEL-ansatz,fig:xxx,fig:tfi,fig:lih} at $\epsilon \rightarrow 0$
		were simulated for an increasing number of qubits and the curve from \cref{eq:fitfunction} was fitted. 
		\label{fig:alpha-n}
	}
\end{figure*}

\subsection{Commutator norm}
\begin{lemma}\label{lemma:commnorm}
	The commutators norms are approximately related as
	\begin{equation}
		\frac{ \lVert [ \rhoid, \rho ] \rVert_1  } {1-\lambda_1}
		= 
		\lVert [ \rhoid,  \rhoerr ] \rVert_1+ \mathcal{E}_q,
	\end{equation}
up to the approximation error $ \mathcal{E}_q$.
\end{lemma}
\begin{proof}
	Using the decomposition from \cref{eq:decomposition} we obtain
	\begin{align*}
		\lVert [ \rhoid, \rho ] \rVert_1  &= 
		\lVert [ \rhoid, \eta \rhoid  ] + [ \rhoid, (1-\eta) \rhoerr ] \rVert_1 \\
		& =
		(1-\eta) \lVert [ \rhoid,  \rhoerr ] \rVert_1 
	\end{align*}
We can approximate
	$\eta = \lambda_1 + \mathcal{O}(\mathcal{E}_{F}) + \mathcal{O}(\mathcal{E}_{C})$ 
	via \cref{eq:fid_approx} and \cref{eq:eigval_appr} and obtain that
	\begin{equation}
		\frac{ \lVert [ \rhoid, \rho ] \rVert_1 } {1-\lambda_1}
		= 
		\lVert [ \rhoid,  \rhoerr ]\rVert_1 + \mathcal{E}_q.
	\end{equation}
The error term can be obtained via the triangle inequality $ |\mathcal{E}_q| \leq  [\mathcal{O}(\mathcal{E}_{F}\mathcal{E}_{C}) + \mathcal{O}(\mathcal{E}_{C}^2)] /(1-\lambda_1)$.
\end{proof}

\section{Further details of numerical simulations \label{app:numerics}}

\subsection{The SEL and HVA ans{\"a}tze\label{app:HVA}}

The circuit structure of the SEL ansatz used in \cref{fig:SEL-ansatz} is illustrated in \cref{fig:ansatz}:
it consists of alternating layers of parametrised single-qubit rotations and a ladder of nearest-neighbour
CNOT gates.

Let us now define the HVA ansatz.
In particular, recall that the HVA ansatz is a discretisation of the adiabatic evolution
\begin{equation*}
	U(\underline{\beta}, \underline{\gamma})
	=
	\prod_{k=1}^\nu
	e^{-i \gamma_k \mathcal{H}_1}
	e^{-i \beta_k \mathcal{H}_0},
\end{equation*}
which is applied to the initial state as the ground state of the trivial Hamiltonian $\mathcal{H}_0$.

The individual evolutions are then trotterised such that a piece of time evolution $e^{-i \gamma_k \mathcal{H}_1}$ is broken up into
products of evolution operators under the individual 
Hamiltonian terms as
\begin{equation*}
	e^{-i \gamma_k \mathcal{H}_1} \rightarrow
	\prod_{l=1}^{r_h} 
	e^{-i \gamma_k h_l P_l}.
\end{equation*}
Above we utilised the decomposition of the non-trivial part of the Hamiltonian $\mathcal{H}_1 = \sum_{l=1}^{r_h} h_l P_l $
into  Pauli strings $P_l \in \{\mathrm{Id}, X, Y, Z \}^{\otimes N}$.

We set the circuit parameter as $\gamma_k = k/\nu$ and $\beta_k = 1 - k/\nu$, such that
the circuit approximates a discretised adiabatic evolution between $\mathcal{H}_0$ and $\mathcal{H}_1$ --
and we will refer to these as VQE parameters.

In the case of random parametrisation of the HVA ansatz, every gate implementing the evolution under a single Pauli string
$e^{-i \gamma_k h_l P_l}$ is assigned a random parameter as $e^{-i \theta_{q} P_l}$
with $|\theta_{q}|  \leq 2\pi$.

\subsection{Inserting additional gates to the TFI ansatz}

In \cref{fig:tfi-hva-scrambled} we repeated the same simulation as in \cref{fig:tfi} (a), i.e., 
using a HVA ansatz for the TFI spin model at random circuit parameters, but we appended to each layer a series of parametrised $Rz$ gates
on each qubit. This guarantees that the dynamic Lie algebra generated by the Pauli terms of the TFI problem in \cref{eq:tfi}
is expanded by the inclusion of Pauli Z operators. Increasing the circuit depth of the HVA ansatz thus leads to
a faster increase of the dimensionality of the Lie algebra which demonstrably leads to a faster scrambling
of local noise into global white noise, e.g., steeper slope of the $\epsilon \rightarrow 0$ fit in \cref{fig:tfi-hva-scrambled}
than in \cref{fig:tfi}.

\subsection{Scaling with the number of qubits}

In \cref{fig:alpha-n} we simulate the same circuits as in \cref{fig:SEL-ansatz,fig:xxx,fig:tfi,fig:lih} at error
rates $\epsilon \rightarrow 0$ and plot the fit parameter $\alpha$---which is the prefactor in \cref{eq:fitfunction}---for
an increasing number of qubits. The results appear to confirm an asymptotically non-increasing trend confirming
theoretical expectations of \cite{dalzell2021random} for random circuits whereby $\alpha$ is constant bounded
in terms of the number of qubits.

%UNCOMMENT THIS TO USE BIB FILE
%\bibliography{references}

%apsrev4-2.bst 2019-01-14 (MD) hand-edited version of apsrev4-1.bst
%Control: key (0)
%Control: author (8) initials jnrlst
%Control: editor formatted (1) identically to author
%Control: production of article title (0) allowed
%Control: page (0) single
%Control: year (1) truncated
%Control: production of eprint (0) enabled
%

\end{document}